\documentclass[aps,twocolumn,showpacs,pra,citeautoscript,amsmath,amssymb,floatfix,superscriptaddress]{revtex4-1}

\usepackage{amsmath}
\usepackage{amssymb}
\usepackage{epsfig}
\usepackage{color}
\usepackage{amsthm}
\usepackage{hyperref}
\usepackage{appendix}

\newtheorem{theorem}{Theorem}
\newtheorem*{theorem*}{Theorem}

\newtheorem{lemma}{Lemma}

\newtheorem{definition}{Definition}
\newtheorem{claim}{Claim}

\newcommand{\ket}[1]{|#1\rangle}

\newcommand{\bracket}[2]{\langle #1|#2\rangle}
\newcommand{\ketbra}[2]{|#1\rangle\langle #2|}
\newcommand{\tr}[0]{\textnormal{Tr}}

\begin{document}
\title{Communication tasks with infinite quantum-classical separation}
\author{Christopher Perry}
\affiliation{Department of Physics and Astronomy, University College London, Gower Street, London WC1E 6BT, United Kingdom}
\author{Rahul Jain}
\affiliation{Department of Computer Science and Centre for Quantum Technologies, National University of Singapore, Singapore 119615}
\affiliation{MajuLab, CNRS-UNS-NUS-NTU International Joint Research Unit, UMI 3654, Singapore}
\author{Jonathan Oppenheim}
\affiliation{Department of Physics and Astronomy, University College London, Gower Street, London WC1E 6BT, United Kingdom}
\affiliation{Department of Computer Science and Centre for Quantum Technologies, National University of Singapore, Singapore 119615}

\begin{abstract}
Quantum resources can be more powerful than classical resources -- a quantum computer can solve certain problems exponentially faster than a classical computer, and computing a function of two people's inputs can be done with exponentially less communication with quantum messages than with classical ones. Here we consider a task between two players, Alice and Bob where quantum resources are infinitely more powerful than their classical counterpart. Alice is given a string of length $n$, and Bob's task is to exclude certain combinations of bits that Alice might have. If Alice must send classical messages, then she must reveal nearly $n$ bits of information to Bob, but if she is allowed to send quantum bits, the amount of information she must reveal goes to zero with increasing $n$. Next, we consider a version of the task where the parties may have access to entanglement. With this assistance, Alice only needs to send a constant number of bits, while without entanglement, the number of bits Alice must send grows linearly with $n$. The task is related to the PBR theorem which arises in the context of the foundations of quantum theory.
\end{abstract}

\maketitle

\emph{Introduction.} In a typical communication task, two players, Alice and Bob, are given inputs $x$ and $y$ and asked to compute some function or relation, $f\left(x,y\right)$. As initially neither player has any knowledge of the other's input, some communication will have to take place between the parties to achieve their goal. Depending on the resources available to them, this communication may involve sending quantum states or perhaps be restricted to sending classical messages. How much of an advantage can be gained in using quantum strategies over classical ones? The standard measure used to investigate this question is the communication complexity \cite{Yao1979}, the minimum amount of bits or qubits the players must exchange to succeed. Tasks exist for which there is an exponential separation between the quantum and classical communication complexities \cite{Raz1999, Buhrman2001, DeWolf2003} and in the absence of shared entanglement, it is known that such a separation is maximal in the bounded error model \cite{Kremer1995}.

Here we consider two modified scenarios and ask how big the separation can be. Firstly, rather than analyzing how much communication is needed in a given task, we look at the amount of information regarding the players' inputs that needs to be exchanged. For our task we find that in the zero error setting, it is possible to have an infinite separation with respect to this measure: classically nearly all of the information needs to be revealed while a quantum strategy can succeed yet reveal next to nothing. This result has clear implications if one is concerned about keeping such information private. If we instead want an infinite separation in the number of sent bits, rather than the amount of sent information, we are able to do so by allowing the players to abort some fraction of the games they play. Here an entanglement assisted strategy has constant communication complexity, while the purely classical complexity is $\Omega(n)$.

The game we consider, and shall refer to as the exclusion game, involves Alice and Bob, together with a referee to mediate the task. It runs as follows. First, the referee gives Alice an $n$-bit string, $\vec{x}\in\{0,1\}^{n}$, with each of the $2^n$ strings being equally likely. Alice is then allowed to send a single message regarding her input to Bob. Next, the referee chooses at random a subset, $y\subseteq[n]$ of size $m$, of locations in Alice's bit string and gives this to Bob. There are ${n \choose m}$ possible subsets and they are all equally likely. If $\mathcal{M}_y\left(\vec{x}\right)$ denotes the $m$-bit string formed by restricting $\vec{x}$ to the bits specified by $y$, Bob's task is to produce a string $\vec{z}_y\in\{0,1\}^{m}$ such that $\mathcal{M}_y\left(\vec{x}\right)\neq\vec{z}_y$.

As an illustration, consider a game where $n=3$, $m=2$ and the inputs given to Alice and Bob are $\vec{x}=001$ and $y=\{1,3\}$ respectively. Winning answers that Bob can give would then be $\vec{z}_y\in\{00,10,11\}$ as the only losing answer is $\vec{z}_y=\mathcal{M}_y\left(\vec{x}\right)=\mathcal{M}_{\{1,3\}}\left(001\right)=01$.

More formally, the amount of information that the players reveal to one another about their inputs is called the internal information cost of the protocol \footnote{For our game, because the players' inputs come from a product distribution, the internal information cost is equal to the external information cost, the amount of information the players reveal to an external observer.}. The information cost is a useful quantity as it lower bounds the protocol's communication complexity \cite{Chakrabarti2001, Bar2002}. In classical information theory, it has found use in proving direct sum theorems \cite{Chakrabarti2001, Bar2002, Jain2003, Barak2013} and while for quantum protocols involving multiple rounds there have been many definitions (see for example \cite{Jain2003b, Jain2010, Braverman2012} and in particular \cite{Touchette2014} for a recent, fully quantum generalization of the classical case), it is relatively simple to define for single round schemes. If one wants to reveal as little information as possible, it is natural to ask if an advantage can be gained in using quantum protocols and there are known exponential separations \cite{Kerenidis2012}. We find that in the exclusion game, for certain choices of $m$, classical strategies must reveal greater than $n-o(n)$ bits of information. Quantum mechanics however, admits a strategy for which the information cost tends to zero in the limit of large $n$.

For our second result, we consider how the communication complexity changes when the players are allowed to share entanglement. This scenario was originally formulated in \cite{Cleve1997} (and developed in \cite{Buhrman1999, Buhrman2001b}) where a task was found for which sharing an entangled state reduces the communication complexity by a single bit. Exponential separations between what is possible with entanglement assisted and classical strategies have also been found \cite{Gavinsky2009, Gavinsky2009b} but in general, it is known that almost all Boolean functions have linear communication complexity even in the presence of shared entanglement \cite{Buhrman2001c,Gavinsky2006,Montanaro2007}. For a recent survey, see \cite{Buhrman2010}. By modifying the exclusion game to allow Alice to decline to play with probability $\delta$, it is possible to find an entanglement assisted scheme for which the communication complexity is less than a constant for particular $m$. For purely classical strategies, the communication complexity is $\Omega(n)$.

Previous unbounded separations for communication tasks exist in the non-deterministic setting. Here, for a Boolean function, two parties are required to compute $f$ correctly with certainty if $f\left(x,y\right)=0$ and non-zero probability if $f\left(x,y\right)=1$. In this regime, a 1 qubit vs. $\log\left(n\right)$ bits separation has been found for the communication complexity \cite{Massar2001} and a 1 vs. $n$ gap exists for the query complexity \cite{DeWolf2003}. Furthermore, \cite{Massar2001} uses this separation to show that unbounded classical communication is needed to simulate bipartite measurements on a Bell state if the parties share only a finite amount of randomness. In a similar vein, it was shown in \cite{Galvao2003} that there exist scenarios where a qubit can be substituted only for an unbounded number of classical bits.

%Both of these scenarios exhibit a larger than exponential gap between what is possible with quantum and classical strategies in communication protocols. The near maximal for the information cost, is comparable to that found in the Deutsch-Jozsa algorithm with respect to the number of oracle calls \cite{Deutsch1992}, non-deterministic query complexity problems \cite{DeWolf2003} and space complexity theory \cite{Galvao2003}.

This paper is organized as follows. First, we investigate the amount of information revealed, giving a quantum strategy and lower bounding the classical information cost. Next we consider the case where Alice is allowed to abort the game. We lower bound the classical communication complexity and give a quantum strategy making use of shared entangled states. Appendices can be found in the supplementary material. Throughout this paper standard notation for asymptotic complexity ($O,o,\Omega,\omega$) is used. Formal definitions can be found in, for example, \cite{Knuth1976}.   

\emph{Information revealed.} More formally, how do we quantify the amount of information revealed in a given task? Let $X$ and $Y$ denote the random variables, distributed according to some joint distribution $\mu$, received by Alice and Bob respectively. Let $\pi$ be the protocol they follow in attempting to achieve their goal and $\pi\left(X,Y\right)$ denote the public randomness and messages exchanged during the protocol. The internal information cost of the protocol is then given by \cite{Barak2013}:
\begin{equation}
IC_{\mu}\left(\pi\right)=I\left(X:\pi\left(X,Y\right)|Y\right)+I\left(Y:\pi\left(X,Y\right)|X\right), \label{Inf Cost}
\end{equation}
where $I(R:T|U)$ denotes the mutual information between $R$ and $T$ given knowledge of $U$. In terms of the Shannon entropy $H(R)$, $I(R:T|U):=H(R,U)+H(T,U)-H(R,T,U)-H(U)$. Intuitively, the first term in Eq. (\ref{Inf Cost}) captures the amount of information Bob gains about Alice's input, $X$, by following the protocol, $\pi$. Conditioning on $Y$ accounts for any correlations that exist between $X$ and $Y$. The second term reverses the roles of Alice and Bob.

In the exclusion task focused on here, Alice and Bob's inputs are uniform and independent of one another. Furthermore, a protocol consists of sending a single message from Alice to Bob. This simplifies Eq. (\ref{Inf Cost}) so that for a classical message, $M_{C}$:
\begin{align}
IC_{\textit{unif}}\left(M_{C}\right)&=n-H(X|M_{C}), \label{Clas Inf Cost}
\end{align}
where $H(T|U)$ is the conditional Shannon entropy. For general $\mu$, when the message is quantum, denoted $M_{Q}$:
\begin{align}
IC_{\mu}\left(M_{Q}\right)&\leq2S(M_{Q}),\label{Quan Inf Cost}
\end{align}
where $S(R)$ is the von Neumann entropy.

To devise a quantum strategy, consider the measurement used by Pusey, Barrett and Rudolph (PBR) in the context of investigating the reality of the quantum state \cite{Pusey2012}. The measurement in question applies to the following scenario. Suppose $r$ systems are each prepared in one of two states:
\begin{align}
\begin{split}
\ket{\psi_0\left(\theta\right)}&=\cos\left(\frac{\theta}{2}\right)\ket{0}+\sin\left(\frac{\theta}{2}\right)\ket{1}, \\
\ket{\psi_1\left(\theta\right)}&=\cos\left(\frac{\theta}{2}\right)\ket{0}-\sin\left(\frac{\theta}{2}\right)\ket{1},
\end{split}
\end{align}
so that in total there are $2^r$ possible preparations:
\begin{equation}
\mathcal{P}=\left\{\ket{\Psi_{\vec{x}}\left(\theta\right)}=\bigotimes^{r}_{i=1}\ket{\psi_{x_i}\left(\theta\right)}\right\}_{\vec{x}\in\{0,1\}^r}.
\end{equation}
PBR noted that if $\theta$ is chosen to be:
\begin{equation}
\theta_r=2\arctan\left(2^{1/r}-1\right), \label{PBR angle}
\end{equation}
it is possible to perform a global measurement across the $r$ systems such that the outcome enables one to deduce a preparation that has not taken place. In other words, if the global preparation resulted in $\ket{\Psi_{\vec{x}}}$, after the measurement it is possible to produce a $\vec{z}$ such that $\vec{z}\neq\vec{x}$ with certainty.

This is the smallest value of $\theta$ for which such a measurement is possible \cite{Bandyopadhyay2014} and the measurement to perform in this case is given by the set of projectors, $\mathcal{M}=\{\ket{\zeta_{\vec{z}}}\}_{\vec{z}\in\{0,1\}^r}$, where:
\begin{equation}
\ket{\zeta_{\vec{z}}}=\frac{1}{\sqrt{2^r}}\left(\ket{\vec{0}}-\sum_{\vec{s}\neq\vec{0}}\left(-1\right)^{\vec{z}\cdot\vec{s}}\ket{\vec{s}}\right). \label{PBR Projector}
\end{equation}

Converting these results into a quantum strategy for playing the exclusion game leads to the following:
\begin{theorem} \label{Quantum IR}
Suppose $m\in\omega\left(n^{\frac{1}{2}+\beta}\right), \beta>0$. Then there exists a quantum strategy for the exclusion game (for all prior distributions on $\vec{x}$ and $y$) such that Bob is able to produce $\vec{z}_y\neq\mathcal{M}_y\left(\vec{x}\right)$, for any $y$, while the amount of information Alice reveals to Bob regarding $\vec{x}$ tends to zero in the limit of large $n$.
\end{theorem}
\begin{proof}
The full proof is given in Appendix B. Suppose that upon receiving the bit string $\vec{x}$ from the referee, Alice prepares the state $\ket{\Psi_{\vec{x}}\left(\theta_m\right)}$ where $\theta_m$ is defined by Eq. (\ref{PBR angle}). She sends this state to Bob. The referee then gives input $y$ to Bob who takes the systems in Alice's message identified by $y$ and performs the measurement described in Eq. (\ref{PBR Projector}). This allows him to produce a $\vec{z}_y$ such that $\vec{z}_y\neq\mathcal{M}_y\left(\vec{x}\right)$ with certainty. Hence Alice and Bob succeed in their task.

To upper bound the amount of information this strategy reveals, by Eq. (\ref{Quan Inf Cost}) it suffices to consider the entropy of the message sent by Alice. That this tends to zero for the $m$ specified, is shown in Appendix B. Essentially, as $n$ increases, the angle between $\ket{\psi_0}$ and $\ket{\psi_1}$ can be made smaller while still allowing exclusion to be possible.
\end{proof}

How much information must Alice reveal to Bob in a classical strategy? For him to succeed with certainty, the message that Alice sends needs to allow him to produce a set of answers, $A_{\vec{x}}=\{\vec{z}_y\}$ such that $\vec{z}_y\neq\mathcal{M}_y\left(\vec{x}\right)$ for each possible $y$. Each of the ${n \choose m}$ elements of $A_{\vec{x}}$ allows Bob to deduce a set, $S_{{\vec{z}}_y}$, of $2^{n-m}$ strings not equal to $\vec{x}$. $S_{{\vec{z}}_y}$ consists of all $\vec{x}$ such that $\mathcal{M}_y\left(\vec{x}\right)=\vec{z}_y$. Hence each $\vec{z}_y\in A_{\vec{x}}$ reveals some information about $\vec{x}$ to Bob, although there may be some overlap between the elements in different $S_{\vec{z}_y}$. To lower bound the amount of information that is revealed we need to find the $A_{\vec{x}}$ that allows Bob to exclude the fewest possible candidates for $\vec{x}$.

Doing so leads to the following result:
\begin{theorem} \label{Classical IR}
If $\vec{x}$ and $y$ are chosen independently and from the uniform distribution, $\mu=\textit{unif}$:
\begin{enumerate}
\item Any classical strategy, $M_C$, for the exclusion game such that Bob is able to produce $\vec{z}_y\neq\mathcal{M}_y\left(\vec{x}\right)$, for any $y$, is such that:
\begin{equation}
\textit{IC}_{\textit{unif}}\left(M_C\right)\geq n - \log_{2} \left(\gamma_m\right), \label{Classical IL general m}
\end{equation}
where $\gamma_m=\sum_{i=0}^{m-1}{n\choose i}$.
\item For the following parameterizations of $m$ we find:
\begin{enumerate}
\item If both $m\in\omega\left(\sqrt{n}\right)$ and $m\in o(n)$ hold, then $\textit{IC}_{\textit{unif}}\left(M_C\right)\geq n-o(n)$.
\item If $m=\alpha n$ for some constant $\alpha$, $0<\alpha<\frac{1}{2}$, then $\textit{IC}_{\textit{unif}}\left(M_C\right)\in\Omega\left(n\right)$.
\end{enumerate}
\end{enumerate}
\end{theorem}
\begin{proof}
The full proof is given in Appendix C. First we show that a set of answers that allows Bob to exclude the fewest possible $\vec{x}$ is of the form $A_{\vec{x}}=\{\vec{z}_y: \vec{z}_y=\mathcal{M}_y\left(\vec{a}_{\vec{x}}\right)\}$ where $\vec{a}_{\vec{x}}\in\{0,1\}^n$ is some suitably chosen bit string such that $\mathcal{M}_{y}\left(\vec{a}_{\vec{x}}\right)\neq\mathcal{M}_{y}\left(\vec{x}\right)$, $\forall y$. Without loss of generality, to calculate the number of strings such a $A_{\vec{x}}$ will exclude, we can assume $\vec{a}_{\vec{x}}$ to be the all zero string, $\vec{0}$. Hence the $\vec{x}$ that Bob can exclude are precisely those containing $m$ or more zeros. The number of remaining possibilities is given by $\gamma_m=\sum_{i=0}^{m-1} {n \choose i}$ and to lower bound the amount of information revealed, it is sufficient to assume that Bob believes that they are all equally likely. Using this fact and Eq. (\ref{Clas Inf Cost}) gives Eq. (\ref{Classical IL general m}).

Part 2 follows by considering the scaling of $\gamma_m$ for the stated $m$. This is done in Appendix C.
\end{proof}

From Theorem \ref{Quantum IR} and Theorem \ref{Classical IR} Part 2a, we obtain our first infinite separation between quantum and classical mechanics. For the exclusion game, there exists a quantum strategy such that for certain choices of $m$, the amount of information Alice must reveal to Bob tends to 0 in the limit of large $n$. On the other hand, for the same scaling of $m$ all classical strategies must reveal nearly $n$ bits of information about $\vec{x}$ to Bob. Quantum mechanics allows Alice to reveal almost nothing about her input while classically she must reveal close to everything.

In the discussion so far, we have demanded that Alice and Bob's strategy should allow Bob to always output a winning string. What impact does allowing Bob to make an error with probability at most $\varepsilon$ have? The scaling given in Theorem \ref{Classical IR} Part 2 is not robust against allowing such an error. To see this, suppose that Alice sends no information to Bob and upon receiving input $y$ from the referee he is forced to guess an answer. There are $2^m$ possible strings he can give and of these only one, that which is equal to $\mathcal{M}_y\left(\vec{x}\right)$, is incorrect. Hence, for $\varepsilon\geq\frac{1}{2^m}$, Alice does not need to send a message to Bob and thus reveals no information regarding $\vec{x}$. 

%For $\varepsilon<\frac{1}{2^m}$, we do not have a complete characterization of how the amount of information revealed behaves with $n$. However, for $m=\sqrt{n}$, it is possible to give a strategy which uses 1 bit of classical communication and achieves an error of at most $\varepsilon=\frac{1}{2^{m+1}}$. This strategy is given in Appendix C.

\emph{Entanglement assisted communication complexity.} It should also be noted that the quantum strategy given requires exactly $n$ qubits to be sent from Alice to Bob while, as the information cost lower bounds the communication cost, an optimal classical strategy may require $n-o(n)$ bits. By modifying the game, we obtain a task which admits a strategy involving entanglement with constant communication complexity while all classical strategies involve at least $\Omega(n)$ bits being sent. In what follows, Alice may choose to abort the game with probability $\delta$ on each pair of inputs, $(\vec{x},y)$, and the players have access to both private and shared randomness. When she does not abort however, Bob must give a correct answer.

How does this change affect the classical communication complexity?
\begin{theorem} \label{LV CC Theorem}
Suppose $m=\alpha n$, $0<\alpha<\frac{1}{2}$ and Alice can abort with probability at most $\delta>0$ on each pair of inputs, $(\vec{x},y)$. Any classical strategy for the exclusion game such that when Alice does not abort, Bob is able to produce $\vec{z}_y\neq\mathcal{M}_y\left(\vec{x}\right)$ for any $y$, has communication cost $\Omega(n)$.
\end{theorem}
\begin{proof}
The full proof is given in Appendix D. First we assume that Alice is allowed to abort with average probability at most $\delta$ where the average is taken over all inputs and any randomness used. We show that in this setting any winning protocol has information cost $\Omega(n)$ when $\vec{x}$ and $y$ are independent and uniformly distributed. Hence, any winning protocol where Alice aborts with probability at most $\delta$ for each pair of inputs also has information cost $\Omega(n)$ on this distribution. Finally, the information cost of a protocol lower bounds its communication cost.
\end{proof}

With access to entangled states, rather than sending $\ket{\Psi_{\vec{x}}\left(\theta_m\right)}$ to Bob directly, Alice could instead attempt to steer Bob's side of the entanglement to the desired state by performing an appropriate measurement on her own system. To see how this would work, suppose Alice and Bob share $n$ entangled states, one for each bit in $\vec{x}$. From \cite{Rudolph2004} we know that there exists an entangled state, $\ket{\Phi}_{AB}$, and two measurements with outcomes labeled by 0 and 1, $\mathcal{S}=\{S_0, S_1\}$ and $\mathcal{R}=\{R_0, R_1\}$, with the following properties. Firstly, if Alice measures her half of $\ket{\Phi}_{AB}$ with $\mathcal{S}$ and obtains the outcome 0, Bob's half of $\ket{\Phi}_{AB}$ is steered to $\ket{\psi_0\left(\theta_m\right)}$ while if she obtains outcome 1, Bob's system is steered to the state $\ket{-}$. Similarly, measuring with $\mathcal{R}$ will steer Bob to either $\ket{\psi_1\left(\theta_m\right)}$ or $\ket{+}$. If the value of $x_i$ determines which of $\mathcal{S}$ and $\mathcal{R}$ Alice applies, the probability that Bob's system is steered to the state $\ket{\psi_{x_i}\left(\theta_m\right)}$ is:
\begin{equation}
P_{\textrm{steer}}=\frac{1}{1+\sin\theta_m}.
\end{equation}
Full details on the form of $\ket{\Phi}_{AB}$, $\mathcal{S}$ and $\mathcal{R}$ are given in Appendix E.

Making use of this steering while allowing Alice to occasionally abort gives the following result.
\begin{theorem} \label{Entanglement Assisted Theorem}
Suppose $m=\alpha n$, $0<\alpha\leq1$ and Alice can abort with probability at most $\delta>0$ on each pair of inputs, $(\vec{x},y)$. Then there exists an entanglement assisted strategy for the exclusion game using $\log_{2}k$ bits of communication, such that when Alice does not abort, Bob is able to produce $\vec{z}_y\neq\mathcal{M}_y\left(\vec{x}\right)$ for any $y$. Here $k$ is some constant that depends on $\delta$ but not on $n$.
\end{theorem}
\begin{proof}
The full proof is given in Appendix E. Suppose Alice and Bob share $k$ sets of $n$ copies of $\ket{\Phi}_{AB}$. On receiving $\vec{x}$, on the $i^{\textit{th}}$ copy in each set, Alice measures her half using $\mathcal{S}$, if $x_i=0$, and $\mathcal{R}$, if $x_i=1$.

If there is a set for which all measurements resulted in the 0 outcome, Alice knows Bob's systems in this set have been steered to $\ket{\Psi\left(\theta_m\right)}$ and she sends $\log_2{k}$ bits to identify this set to him. Upon receiving $y$, if Bob performs the PBR measurement described by Eq. (\ref{PBR Projector}) on the set, he will obtain $\vec{z}_y\neq\mathcal{M}_y\left(\vec{x}\right)$.

If for each of the $k$ sets the measurement outcome 1 occurs at least once, Alice aborts the game. This happens with probability:
\begin{equation}
P_{abort}=\left(1-\left(\frac{1}{1+\sin\theta_m}\right)^n\right)^k,
\end{equation}
and we show in the Appendix that this can be made less than $\delta$, $\forall n$, for some constant $k$ that depends on $\delta$ but not on $n$.
\end{proof}

From Theorem \ref{LV CC Theorem} and Theorem \ref{Entanglement Assisted Theorem} we obtain our second result. By allowing Alice to occasionally decline to answer, there exist choices of $m$ such that in the exclusion game, with access to entanglement, only a constant amount of communication is required. For classical strategies on the other hand, Alice needs to send $\Omega(n)$ bits of communication.

There is a relation between our two results. It was shown in \cite{Jain2005} that if the first round of a quantum protocol reveals $c$ bits of information, then this round can be replaced with one making use of shared entanglement and $O(c)$ bits of classical communication at the expense of introducing some small additional probability of error. In our setting, allowing the players to abort some fraction of the games avoids this extra error.

\emph{Conclusion.} In this paper, we have designed a communication task which exploits a result from the foundations of quantum mechanics, the PBR theorem. Quantum strategies for this task can drastically outperform classical ones with respect to the amount of information they reveal. Additionally, when the players are allowed an abort probability, the communication complexity is similarly improved by using shared entanglement. In fact, we have shown that while classically on the order of $n$ bits of information need to be revealed or sent, quantum mechanics admits strategies where a constant or even vanishingly small amount is required. This contrasts sharply with the usual measure studied in communication tasks, the communication complexity in the absence of entanglement, where at most an exponential advantage can be gained from using quantum mechanics over classical schemes in the bounded error setting.

Our quantum advantage has been shown to hold for particular scalings of $m$ and in the zero error regime. Open questions are to determine: firstly the optimal quantum strategies in both scenarios for general $m$ and secondly, whether the game can be made more robust against allowing some error. Furthermore, while it can be shown that a quantum strategy for $m\in o(n)$ has communication complexity at least $\Omega\left(\log n\right)$, it is not known if the same holds for $m=\alpha n$ and it may be possible to find a beyond exponential gap with respect to the communication complexity in this regime. Finally, the PBR measurement we use also appears in the task of quantum bet hedging \cite{Arunachalam2013}. It would be interesting to investigate the relationship between the two tasks.

What does the existence of these infinite separations tell us about the structure and power of quantum mechanics? Even though a quantum message may convey a vanishingly small amount of information, to reproduce this information using purely classical means can require an infinitely large amount of information to be sent. The amount of excess informational baggage that a classical model of quantum theory needs to carry round can be very heavy indeed.

\emph{Acknowledgments.} The authors would like to thank Scott Aaronson, Harry Buhrman, Hartmut Klauck, Noah Linden and Matthew Pusey for helpful comments and useful discussions and Matt in particular for bringing the existence of reference \cite{Rudolph2004} to our attention. Part of this work took place while C.P. and J.O. were visiting the `Mathematical Challenges in Quantum Information' programme held at the Isaac Newton Institute for Mathematical Sciences. The work of R.J. is supported by the Singapore Ministry of Education Tier 3 Grant and also the Core Grants of the Centre for Quantum Technologies, Singapore. J.O. is supported by an EPSRC Established Career fellowship.

\bibliography{References}
\bibliographystyle{apsrev4-1}

\clearpage
\widetext
\appendix

\section{Preliminaries} \label{Prelims}

\subsection{Information Theory}

Here we define quantities from both classical and quantum information theory that will be used to define the complexity measures of communication protocols. For a more thorough overview, see, for example, \cite{Nielsen2010}. 

\begin{definition}{\textbf{Entropy.}}
\begin{itemize}
\item The Shannon Entropy of a classical random variable $X$ is given by:
\begin{equation}
H\left(X\right)=-\sum_x p\left(x\right)\log_2p\left(x\right).
\end{equation}
Note in particular that if $X$ has support on $|\mathcal{X}|$ elements, then $H\left(X\right)\leq\log_2{|\mathcal{X}|}$ with equality iff $X$ is uniformly distributed over $\mathcal{X}$.
\item The classical conditional entropy is given by:
\begin{equation} \label{Class Cond Ent}
H\left(X|Y\right)=\sum_y p\left(y\right)H\left(X|Y=y\right),\\
\end{equation}
or equivalently:
\begin{equation}
H\left(X|Y\right)=H\left(X,Y\right)-H\left(Y\right).
\end{equation}
\item The von Neumann entropy of a quantum state $\rho$ is given by:
\begin{equation}
S\left(\rho\right)=-\tr\left(\rho\log_2\rho\right).
\end{equation}
\end{itemize}
\end{definition}

\begin{definition}{\textbf{Mutual Information.}}
\begin{itemize}
\item The mutual information of two classical random variables, $X$ and $Y$ is given by:
\begin{equation}
I\left(X:Y\right)=H\left(X\right)+H\left(Y\right)-H\left(X,Y\right),
\end{equation}
or equivalently:
\begin{equation}
I\left(X:Y\right)=H\left(X\right)-H\left(X|Y\right),
\end{equation}
\item The mutual information between classical random variables $X$ and $Y$ conditioned on a third variable $Z$ is given by:
\begin{equation}
I\left(X:Y|Z\right)=H\left(X|Z\right)-H\left(X|Y,Z\right).
\end{equation}
\item The quantum mutual information for a composite quantum system on $AB$ is given by:
\begin{equation}
I\left(A:B\right)=S\left(A\right)+S\left(B\right)-S\left(A,B\right).
\end{equation}
\item The quantum conditional mutual information for a composite quantum system on $ABC$ (conditioning on subsystem $C$) is given by:
\begin{equation}
I\left(A:B|C\right)=S\left(A,C\right)+S\left(B,C\right)-S\left(C\right)-S\left(A,B,C\right).
\end{equation}
\end{itemize}
\end{definition}

\subsection{Classical Communication Protocols}

In a typical communication task, Alice and Bob are given inputs, $x\in\mathcal{X}$ and $y\in\mathcal{Y}$ respectively, and asked to output $z\in\mathcal{Z}$ such that it matches some function or relation on $x$ and $y$, $f\left(x,y\right)$. As initially neither player has any knowledge of the other's input, some communication, either classical or quantum, may have to take place according to a protocol, $\pi$, for them to achieve their goal. To assist them with their task, they may have access to additional resources such as randomness in the form of either a private or public coin or in the quantum setting, they may share an entangled state.

We will be interested in two measures of the cost of a protocol: the communication cost and the internal information cost and we define these below. For the exclusion game ($\textit{EG}$) considered in this paper, the communication is restricted to be one-way; a message can only be sent from Alice to Bob and our definitions reflect this. We shall use $M_C$ to denote a protocol that requires only a classical message to be sent from Alice to Bob and $M_Q$ to denote one that requires only a quantum message from Alice to Bob. Usually, we also require that upon receiving this message Bob never makes an error in evaluating $\textit{EG}\left(\vec{x},y\right)$.

\begin{definition}{\textbf{One-way, zero error, communication complexity.}}
\begin{itemize}
\item Classical:
\begin{itemize}
\item Let $\Pi_{C}^{\mathit{pub},\rightarrow}$ be the set of protocols in which Alice and Bob have access to shared randomness and their communication is restricted to sending classical messages from Alice to Bob. The communication cost of such a protocol, $\textit{CCC}\left(M_C\right)$, is the maximum number of bits that Alice sends in any run of the protocol.
\item The one-way, zero error, classical communication complexity of $f$, $\textit{CCC}^{\rightarrow}_{0}\left(f\right)$, is the communication cost of the cheapest  protocol in $\Pi_{C}^{\mathit{pub},\rightarrow}$ that allows Bob to evaluate $f\left(x,y\right)$ with zero error.
\end{itemize}
\item Quantum:
\begin{itemize}
\item Let $\Pi_{Q}^{\rightarrow}$ be the set of protocols in which Alice and Bob do not share any entangled states and their communication is restricted to sending quantum messages from Alice to Bob. The communication cost such a protocol, $\textit{QCC}\left(M_Q\right)$, is the maximum number of qubits that Alice sends in any run of the protocol.
\item The one-way, zero error, quantum communication complexity of $f$, $\textit{QCC}^{\rightarrow}_{0}\left(f\right)$, is the communication cost of the cheapest  protocol in $\Pi_{Q}^{\rightarrow}$ that allows Bob to evaluate $f\left(x,y\right)$ with zero error.
\end{itemize}
\item Entanglement Assisted:
\begin{itemize}
\item Let $\Pi_{C}^{\mathit{ent},\rightarrow}$ be the set of protocols in which Alice and Bob may share entangled states and their communication is restricted to sending classical messages from Alice to Bob. The communication cost of such a protocol, $\textit{ECC}\left(M_C\right)$, is the maximum number of bits that Alice sends in any run of the protocol.
\item The one-way, zero error, entanglement assisted communication complexity of $f$, $\textit{ECC}^{\rightarrow}_{0}\left(f\right)$, is the communication cost of the cheapest  protocol in $\Pi_{C}^{\mathit{ent},\rightarrow}$ that allows Bob to evaluate $f\left(x,y\right)$ with zero error.
\end{itemize}
\end{itemize}
\end{definition}

\begin{definition}{\textbf{Internal information cost.}}
Suppose $X$ and $Y$ are distributed according to some joint distribution $\mu$. For protocol $\pi$, let $\pi\left(X,Y\right)$ denote the public randomness and messages exchanged during the protocol. The internal information cost of $\pi$ is then:
\begin{equation}
\textit{IC}_{\mu}\left(\pi\right)=I\left(X:\pi\left(X,Y\right)|Y\right)+I\left(Y:\pi\left(X,Y\right)|X\right).
\end{equation}
If $\pi$ only involves communication from Alice to Bob, this reduces to:
\begin{equation}
\textit{IC}_{\mu}\left(\pi\right)=I\left(X:\pi\left(X\right)|Y\right).
\end{equation}
\end{definition}

In what follows, we will make use of the following lemmas:

\begin{lemma} \label{IC<CC}
For a classical protocol, $\pi_C$, for any distribution, $\mu$, on the inputs:
\begin{equation}
\textit{IC}_{\mu}\left(\pi_C\right)\leq \textit{CCC}\left(\pi_C\right)
\end{equation}
\end{lemma}
\begin{proof}
See, for example, \cite{Braverman2011}.
\end{proof}

\begin{lemma}
Suppose $X$ and $Y$ are independent and uniformly distributed, $\mu=\textit{unif}$. Then in a one-way classical protocol:
\begin{equation}
\textit{IC}_{\textit{unif}}\left(M_C\right)=\log_2\left|\mathcal{X}\right|-H\left(X|M_C\right). \label{Clas Inf Cost Ap}
\end{equation} 
\end{lemma}
\begin{proof}
\begin{align*}
\textit{IC}_{\textit{unif}}\left(M_C\right)&=I\left(X:M_C|Y\right), \nonumber\\
&=H\left(X|Y\right)-H\left(X|M_C,Y\right), \nonumber\\
&=H\left(X\right)-H\left(X|M_C,Y\right), \quad \textrm{as }X\textrm{ is independent of }Y,\nonumber\\
&=H\left(X\right)+H\left(M_C,Y\right)-H\left(X,Y,M_C\right), \nonumber\\
&=H\left(X\right)+H\left(M_C\right)+H\left(Y\right)-H\left(X,M_C\right)-H\left(Y\right), \quad \textrm{as } X, M_C \textrm{ independent of }Y,\nonumber\\
&=H\left(X\right)-H\left(X|M_C\right), \nonumber\\
&=\log_2\left|\mathcal{X}\right|-H\left(X|M_C\right), \quad \textrm{as }X\textrm{ is uniformly distributed.}
%\textit{IC}_{\textit{unif}}\left(M_C\right)&=I\left(X:M_C|Y\right), \nonumber\\
%&=I\left(X:M_C\right), \quad \textrm{as }X\textrm{ is independent of }Y, \nonumber\\
%&=H\left(X\right)-H\left(X|M_C\right), \nonumber\\
%&=\log_2\left|\mathcal{X}\right|-H\left(X|M_C\right), \quad \textrm{as }X\textrm{ is uniformly distributed.} 
\end{align*}
\end{proof}

\begin{lemma}
The information cost of a one-way quantum protocol can be bounded to give:
\begin{equation}
\textit{IC}_{\mu}\left(M_Q\right)\leq 2 S\left(M_Q\right).\label{Quan Inf Cost Ap}
\end{equation}
\end{lemma}
\begin{proof}
\begin{align}
\textit{IC}_{\mu}\left(M_Q\right)&=I\left(X:M_Q|Y\right), \nonumber\\
&=S\left(X,Y\right)+S\left(M_Q,Y\right)-S\left(Y\right)-S\left(X,M_Q,Y\right), \nonumber\\
&=S(X)+S(M_Q)-S(X,M_Q), \quad \textrm{as } X, M_Q \textrm{ independent of }Y,\nonumber\\
&\leq S(X)+S(M_Q) -\left|S(X)-S(M_Q)\right|, \quad \textrm{using the Araki-Lieb inequality for the joint entropy,}\nonumber\\
&\leq2S(M_Q). \nonumber
\end{align}
\end{proof}

\subsubsection{Protocols with abort}

We shall also be interested in a modification of the exclusion game where Alice is allowed to abort the game with some probability. When she does not abort, Bob should again be able to evaluate $f\left(x,y\right)$ with zero error.

\begin{definition}{\textbf{One-way, $\delta$-abort, communication complexity.}}
Given $X,Y$ drawn according to some distribution $\mu$ and a parameter $\delta$, $0<\delta<1$,  define:
\begin{itemize}
\item Classical:
\begin{itemize}
\item The set $\Pi_{C,(\delta,\textit{max})}^{\mathit{pub},\rightarrow}$ to be the set of all classical, one-way, protocols with access to shared randomness such that Alice aborts with probability at most $\delta$ on any pair of inputs, $(x,y)$, and Bob calculates $f(x,y)$ with zero error when she does not abort.
\item $\textit{CCC}^{\rightarrow}_{\delta}\left(f\right)$ to be the classical communication cost of the cheapest protocol contained in $\Pi_{C,(\delta,\textit{max})}^{\mathit{pub},\rightarrow}$. Denote this protocol by $\pi^{*}$.
\item The set $\Pi_{C,(\delta,\mu)}^{\mathit{pub},\rightarrow}$ to be the set of all classical, one-way, protocols with access to shared randomness such that Alice aborts with probability at most $\delta$ (where the inputs are sampled according to $\mu$) and Bob calculates $f(x,y)$ with zero error when she does not abort.
\end{itemize}
\item Entanglement Assisted:
\begin{itemize}
\item The set $\Pi_{C,(\delta,\textit{max})}^{\mathit{ent},\rightarrow}$ be the set of protocols in which Alice and Bob may share entangled states and their communication is restricted to sending classical messages from Alice to Bob such that Alice aborts with probability at most $\delta$ on any pair of inputs, $(x,y)$, and Bob calculates $f(x,y)$ with zero error when she does not abort.
\item $\textit{ECC}^{\rightarrow}_{\delta}\left(f\right)$ to be the entanglement assisted communication cost of the cheapest protocol contained in $\Pi_{C,(\delta,\textit{max})}^{\mathit{ent},\rightarrow}$.
\end{itemize}
\end{itemize}
\end{definition}

The following Lemma will prove useful.

\begin{lemma} \label{Abort Lemma}
For $\delta$, such that $0<\delta<1$, and any distribution, $\mu$, on $X$ and $Y$:
\begin{equation}
\textit{CCC}^{\rightarrow}_{\delta}\left(f\right)=\textit{CCC}\left(\pi^{*}\right)\geq \textit{IC}_{\mu}\left(\pi^{*}\right)\geq\min_{\pi\in\Pi_{C,(\delta,\mu)}^{\mathit{pub},\rightarrow}}\textit{IC}_{\mu}\left(\pi\right).
\end{equation}
\end{lemma}
\begin{proof}
The first inequality follows from Lemma \ref{IC<CC}. To see the second inequality, note that the probability a protocol in $\Pi_{C,(\delta,\textit{max})}^{\mathit{pub},\rightarrow}$ aborts when $X,Y$ are distributed according to $\mu$ is:
\begin{equation}
\sum_{x,y}p\left(\textrm{abort}|x,y\right)\mu\left(x,y\right)\leq\sum_{x,y}\delta\mu\left(x,y\right)\leq\delta.
\end{equation}
Hence $\pi^{*}\in\Pi_{C,(\delta,\textit{max})}^{\mathit{pub},\rightarrow}\subseteq\Pi_{C,(\delta,\mu)}^{\mathit{pub},\rightarrow}$.
\end{proof}

\newpage
\section{Proof of Theorem 1}

\begin{theorem*} \nonumber%\label{Quantum IR}
Suppose $m\in\omega\left(n^{\frac{1}{2}+\beta}\right), \beta>0$. There exists a one-way, quantum strategy for the exclusion game (for all prior distributions on $\vec{x}$ and $y$) such that Bob is able to produce $\vec{z}_y\neq\mathcal{M}_y\left(\vec{x}\right)$, for any $y$, while the amount of information Alice reveals to Bob regarding $\vec{x}$ tends to zero in the limit of large $n$.
\end{theorem*}

More formally, for $m\in\omega\left(n^{\frac{1}{2}+\beta}\right), \beta>0$, there exists a quantum strategy, $M_Q$, such that $\textit{IC}_{\mu}\left(M_Q\right)\rightarrow 0$ as $n\rightarrow \infty$, $\forall \mu$.

\begin{proof}
We give an explicit protocol that achieves this:
\begin{enumerate}
\item Alice receives input $\vec{x}\in\{0,1\}^n$ from the referee.
\item Alice prepares the state:
\begin{equation}
\ket{\Psi_{\vec{x}}\left(\theta_m\right)}=\bigotimes^{n}_{i=1}\ket{\psi_{x_i}\left(\theta_m\right)},
\end{equation}
where:
\begin{align}
\begin{split}
\ket{\psi_0\left(\theta\right)}&=\cos\left(\frac{\theta}{2}\right)\ket{0}+\sin\left(\frac{\theta}{2}\right)\ket{1}, \\
\ket{\psi_1\left(\theta\right)}&=\cos\left(\frac{\theta}{2}\right)\ket{0}-\sin\left(\frac{\theta}{2}\right)\ket{1},
\end{split}
\end{align}
and:
\begin{equation}
\theta_m=2\arctan\left(2^{1/m}-1\right). \label{PBR angle Ap}
\end{equation}
\item Alice sends $M_Q=\ket{\Psi_{\vec{x}}\left(\theta_m\right)}$ to Bob.
\item Bob receives input $y$ from the referee and considers the systems in Alice's message specified by $y$. On these systems, Bob has the state:
\begin{equation}
\ket{\Psi_{\mathcal{M}_y\left(\vec{x}\right)}\left(\theta_m\right)}=\bigotimes_{i\in y}\ket{\psi_{x_i}\left(\theta_m\right)}.
\end{equation}  
\item To the systems specified by $y$, Bob applies the projective measurement, $\mathcal{M}=\{\ket{\zeta_{\vec{z}_y}}\}_{\vec{z}_y\in\{0,1\}^m}$, where:
\begin{equation}
\ket{\zeta_{\vec{z}_y}}=\frac{1}{\sqrt{2^m}}\left(\ket{\vec{0}}-\sum_{\vec{s}\neq\vec{0}}\left(-1\right)^{\vec{z}_y\cdot\vec{s}}\ket{\vec{s}}\right).
\end{equation}
and obtains outcome $\vec{z}_y$.
\item Bob outputs $\vec{z}_y$ as the answer to the referee's question.
\end{enumerate}

Firstly, note that this is a winning strategy as $\bracket{\zeta_{\mathcal{M}_y\left(\vec{x}\right)}}{\Psi_{\mathcal{M}_y\left(\vec{x}\right)}\left(\theta_m\right)}=0$, \cite{Bandyopadhyay2014}, so Bob always outputs $\vec{z}_y\neq\mathcal{M}_y\left(\vec{x}\right)$ and the players always succeed at their task.

To bound the information cost of the protocol, by Eq. (\ref{Quan Inf Cost Ap}), it suffices to consider the entropy of the message sent by Alice:
\begin{align*}
\quad S\left(M_Q\right)&\leq nS\left(\frac{1}{2}\ketbra{\psi_{0}\left(\theta_m\right)}{\psi_{0}\left(\theta_m\right)}+\frac{1}{2}\ketbra{\psi_{1}\left(\theta_m\right)}{\psi_{1}\left(\theta_m\right)}\right),\\
&=n\left[-\left(\left[\cos^2\left(\frac{\theta_m}{2}\right)\right]\log_2\left[\cos^2\left(\frac{\theta_m}{2}\right)\right]+\left[\sin^2\left(\frac{\theta_m}{2}\right)\right]\log_2\left[\sin^2\left(\frac{\theta_m}{2}\right)\right]\right)\right],\\
&< n \left(\frac{\theta_m}{2}\right)^2\left(\frac{1}{\ln 2}-\log_2\left[\left(\frac{\theta_m}{2}\right)^2\right]\right), \quad \textrm{for small }\theta_m.
\end{align*}
Now consider the scaling behavior of $\theta_m$. From Eq. (\ref{PBR angle Ap}), we have:
\begin{equation*}
\frac{1}{m}=\log_2\left(1+\tan\left(\frac{\theta_m}{2}\right)\right).
\end{equation*}
Taking the Taylor series expansion about $\theta_m=0$ gives:
\begin{equation*}
\frac{1}{m}=\frac{1}{\ln2}\frac{\theta_m}{2}-\frac{1}{\ln4}\left(\frac{\theta_m}{2}\right)^2+\frac{{\theta_m}^3}{4\ln8}+O\left({\theta_m}^4\right).
\end{equation*}
Hence, for small $\theta_m$ we have:
\begin{align*}
\frac{1}{m}&=\log_2\left(1+\tan\left(\frac{\theta_m}{2}\right)\right),\\
&< \frac{1}{\ln2}\frac{\theta_m}{2},\\
&<\frac{2}{\ln2}\frac{\theta_m}{2}-\frac{2}{\ln4}\left(\frac{\theta_m}{2}\right)^2,\\
&< \frac{2}{m}.
\end{align*}
%Now consider the scaling behavior of $\theta_m$. From Eq. (\ref{PBR angle Ap}), we have:
%\begin{align*}
%\frac{1}{m}&=\log_2\left(1+\tan\left(\frac{\theta_m}{2}\right)\right),\\
%&< \frac{1}{\ln2}\frac{\theta_m}{2}, \quad \textrm{for small }\theta_m,\\
%&< \frac{2}{m}.
%\end{align*}
Using these upper and lower bounds on $\theta_m$, we obtain:
\begin{align*}
S\left(M_Q\right)<\frac{n}{m^2}\left(2\ln2\right)^2\left[\frac{1}{\ln2}+\log_2\left(\frac{m^2}{\left(\ln2\right)^2}\right)\right], \quad \textrm{for large }m.
\end{align*}
Hence, provided $m\in\omega\left(n^{\frac{1}{2}+\beta}\right), \beta>0$, the entropy of the message sent by Alice and the information complexity of the protocol, tend to zero in the limit of large $n$.
\end{proof}

\newpage
\section{Proof of Theorem 2}

\begin{theorem*} 
If $\vec{x}$ and $y$ are chosen independently and from the uniform distribution, $\mu=\textit{unif}$:
\begin{enumerate}
\item Any one-way, public coin, classical strategy, $M_C$, for the exclusion game such that Bob is able to produce $\vec{z}_y\neq\mathcal{M}_y\left(\vec{x}\right)$, for any $y$, is such that:
\begin{equation}
\textit{IC}_{\textit{unif}}\left(M_C\right)\geq n - \log_{2} \left(\gamma_m\right), \label{Classical IL general m Ap}
\end{equation}
where $\gamma_m=\sum_{i=0}^{m-1}{n\choose i}$.
\item For the following paramaterizations of $m$ we find:
\begin{enumerate}
\item If both $m\in\omega\left(\sqrt{n}\right)$ and $m\in o(n)$ hold, then $\textit{IC}_{\textit{unif}}\left(M_C\right)\geq n-o(n)$.
\item If $m=\alpha n$ for some constant $\alpha$, $0<\alpha<\frac{1}{2}$, then $\textit{IC}_{\textit{unif}}\left(M_C\right)\in\Omega\left(n\right)$.
\end{enumerate}
\end{enumerate}
\end{theorem*}

\subsection{Proof of Part 1.}
Since Bob has to answer correctly with probability one, we can assume that Bob's strategy is deterministic (by fixing Bob's private coins). Recall that $M_C$ includes the public coins of the protocol and note that Alice is allowed to use private coins. 

For any winning strategy, upon receiving $M_C$ from Alice, Bob must be able to construct a correct answer, $\vec{z}_y\neq\mathcal{M}_y\left(\vec{x}\right)$ for each possible $y$. We denote this set of answers by $A_{\vec{x}}=\{\vec{z}_y\}$. Each of the ${n \choose m}$ elements of $A_{\vec{x}}$ allows Bob to deduce a set, $S_{{\vec{z}}_y}$, of $2^{n-m}$ strings not equal to Alice's input, $\vec{x}$. $S_{{\vec{z}}_y}$ consists of all $\vec{x}$ such that $\mathcal{M}_y\left(\vec{x}\right)=\vec{z}_y$. Hence each $\vec{z}_y\in A_{\vec{x}}$ reveals some information about Alice's input to Bob, although there may be some overlap between the elements in different $S_{\vec{z}_y}$. The complete set of strings that $M_C$ allows Bob to rule out is given by $S\left(M_C\right)=\cup_{y}S_{\vec{z}_y}$.

Let $T\left(M_C\right)$ be the set of $\vec{x}$ that the message $M_C$ does not allow Bob to rule out. Then $H\left(X|M_C\right)$ will have support only on $\vec{x}\in T\left(M_C\right)$ and hence $H\left(X|M_C\right)\leq\log_{2}\left|T\left(M_C\right)\right|.$ Combining this with Eq. (\ref{Clas Inf Cost Ap}) gives:
\begin{equation}
\textit{IC}_{\textit{unif}}\left(M_C\right)\geq n -\log_{2}\left|T\left(M_C\right)\right|. \label{Support Bound}
\end{equation}
To lower bound the information cost of a winning protocol, we need to calculate the set of answers which allows Bob to exclude the fewest possible strings.

\begin{claim} \nonumber
The set $A_{\vec{x}}$ that minimizes the size of $S\left(M_C\right)=\cup_{\vec{z}_{y}}S_{\vec{z}_{y}}$ is of the form $A_{\vec{x}}=\{\vec{z}_y: \vec{z}_y=\mathcal{M}_y\left(\vec{a}_{\vec{x}}\right)\}$, where $\vec{a}_{\vec{x}}\in\{0,1\}^n$ is some suitably chosen bit string such that $\mathcal{M}_{y}\left(\vec{a}_{\vec{x}}\right)\neq\mathcal{M}_{y}\left(\vec{x}\right)$, $\forall y$.
\end{claim}
\begin{proof}
To determine the set of answers, $A_{\vec{x}}$, which minimizes the size of $S=\cup_{\vec{z}_{y}}S_{\vec{z}_{y}}$, first:
\begin{itemize}
\item Label the answers $\vec{z}_{y_i}$, $1\leq i\leq{n \choose m}$.
\item Let $k_{ij}=|y_i\cap y_j|$ be the number of places in which answers $\vec{z}_{y_i}$ and $\vec{z}_{y_j}$ overlap (i.e. refer to the same bit in $\vec{x}$). Note that $0\leq k_{ij} \leq m-1$.
\item Similarly define $k_{ij\dots l}$ to be the number of places where answers $\vec{z}_{y_i},\vec{z}_{y_j},\dots \vec{z}_{y_l}$ overlap.
\item Let $r_{ij}$ be the number of places in which answers $\vec{z}_{y_i}$ and $\vec{z}_{y_j}$ agree (i.e. assign the same value to a common location in $\vec{x}$). Note that $0\leq r_{ij} \leq k_{ij}$. 
\end{itemize}
With these definitions, we proceed as follows:
\begin{itemize}
\item Answer $\vec{z}_{y_1}$ excludes $2^{n-m}$ strings.
\item Answer $\vec{z}_{y_2}$ excludes $2^{n-m}$ strings. Some of these strings may have already been excluded by $\vec{z}_{y_1}$ and this will occur iff $r_{12}=k_{12}$, i.e. the two answers give the same value for the bits they overlap on. The number of strings that have already been excluded by $\vec{z}_1$ is then $\delta_{r_{12},k_{12}}2^{n-2m+k_{12}}$, so the number of new strings excluded by $\vec{z}_{y_2}$ is:
\begin{equation}
2^{n-m}-\delta_{r_{12},k_{12}}2^{n-2m+k_{12}}.
\end{equation}
\item Answer $\vec{z}_{y_3}$ excludes $2^{n-m}$ strings but we need to subtract the strings excluded by ($\vec{z}_{y_1}$ and $\vec{z}_{y_3}$), ($\vec{z}_{y_2}$ and $\vec{z}_{y_3}$) and add back in the strings excluded by ($\vec{z}_{y_1}$ and $\vec{z}_{y_2}$ and $\vec{z}_{y_3}$). The number of new strings excluded is thus given by:
\begin{equation}
2^{n-m}-\delta_{r_{13},k_{13}}2^{n-2m+k_{13}}-\delta_{r_{23},k_{23}}2^{n-2m+k_{23}}+\delta_{r_{12},k_{12}}\delta_{r_{13},k_{13}}\delta_{r_{23},k_{23}}2^{n-3m+k_{12}+k_{13}+k_{23}-k_{123}}. \label{Ans Construct}
\end{equation}
Here $k_{123}$ is the number of locations where $\vec{z}_{y_1}$, $\vec{z}_{y_2}$ and $\vec{z}_{y_3}$ overlap.
\item This construction then needs to be continued up to answer $\vec{z}_{y_{n \choose m}}$ and the number of new strings each mask excludes summed to give the total number of strings excluded.
\end{itemize}
From this construction, we see that to minimize the number of strings excluded, one way is to choose $A_{\vec{x}}$ to be such that $r_{ij}=k_{ij}, \forall i,j$. Note that if we had $r_{13}<k_{13}$ in Eq. (\ref{Ans Construct}), then it is not possible to exclude fewer strings. To see this note that for three subsets $y_1, y_2$ and $y_3$ of $[n]$, each of size $m$:
\begin{align*}
m&=|y_2|,\\
&\geq|y_2\cap\left(y_1\cup y_3\right)|,\\
&=|y_1\cap y_2|+|y_2\cap y_3|-|y_1\cap y_2 \cap y_3|,\\
&=k_{12}+k_{23}-k_{123},\\
\Rightarrow\quad n-3m+k_{12}+k_{13}+k_{23}-k_{123}&\leq n-2m+k_{13},\\
\Rightarrow\quad 2^{n-3m+k_{12}+k_{13}+k_{23}-k_{123}}&\leq 2^{n-2m+k_{13}},\\
\Rightarrow\quad 2^{n-2m+k_{23}}&\leq 2^{n-2m+k_{13}}+2^{n-2m+k_{23}}-2^{n-3m+k_{12}+k_{13}+k_{23}-k_{123}},
\end{align*}
and setting $\delta_{r_{13},k_{13}}=0$ does not exclude fewer strings. Similar arguments show that other $\delta_{r,k}$ must be non zero.

Hence, the answers should be consistent with one another i.e. $A_{\vec{x}}=\{\vec{z}_y: \vec{z}_y=\mathcal{M}_y\left(\vec{a}_{\vec{x}}\right)\}$ where $\vec{a}_{\vec{x}}\in\{0,1\}^n$ is some suitably chosen bit string that ensures the $\vec{z}_y$ are winning answers.
\end{proof}

Without loss of generality, to calculate the number of strings such a $A_{\vec{x}}$ will exclude, we can assume $\vec{a}_{\vec{x}}$ to be the all zero string, $\vec{0}$. Here the $\vec{x}$ that Bob can exclude are precisely those containing $m$ or more zeros. The number of remaining possibilities is given by $\gamma_m=\sum_{i=0}^{m-1} {n \choose i}$. Substituting $\left|T\left(M_C\right)\right|=\gamma_m$ in Eq. (\ref{Support Bound}), gives:
\begin{equation}
\textit{IC}_{\textit{unif}}\left(\pi\right)\geq n - \log_{2} \left(\gamma_m\right),
\end{equation}
as required.

\subsection{Proof of Part 2.}

Given Eq. (\ref{Classical IL general m Ap}), we wish to show how it behaves for the particular $m$ given in the statement of Part 2. To do this the following Lemma will be useful:
\begin{lemma} \label{Comb Sum Lemma} \cite[Page 427]{Flum2006}
Let  $n\geq1$ and $0<q\leq\frac{1}{2}$. Then:
\begin{equation}
\sum_{i=0}^{\left\lfloor qn\right\rfloor}{n \choose i}\leq 2^{nH\left(q\right)},
\end{equation}
where $H(q)$ is the binary entropy of $q$.
\end{lemma}

\subsubsection{Instance (a).}

Here both $m\in\omega\left(\sqrt{n}\right)$ and $m\in o(n)$ hold. Suppose $m=n^{1-\epsilon}$, $0<\epsilon<\frac{1}{2}$. Then:
\begin{align*}
\textit{IC}_{\textit{unif}}\left(M_C\right)&\geq n-\log_2\gamma_{n^{1-\epsilon}},\\
&> n-\log_2\left(\sum_{i=0}^{n^{1-\epsilon}} {n \choose i}\right),\\
&\geq n -\log_2\left(2^{n H\left(n^{-\epsilon}\right)}\right), \quad \textrm{(Using Lemma \ref{Comb Sum Lemma})},\\
&=n-n H\left(n^{-\epsilon}\right),\\
&\geq n -\log_2\left(e\right)n^{1-\epsilon}-\epsilon n^{1-\epsilon}\log_2\left(n\right) \quad \textrm{for large n}.
\end{align*}
Hence for this parametrization of $m$, $\textit{IC}_{\textit{unif}}\left(M_C\right)\geq n-o(n)$.

\subsubsection{Instance (b).}

Here $m=\alpha n$, for some constant $\alpha$ such that $0<\alpha<\frac{1}{2}$. Then:
\begin{align*}
\textit{IC}_{\textit{unif}}\left(M_C\right)&\geq n-\log_2\gamma_{\alpha n},\\
&>n-\log_2\left(\sum_{i=0}^{\alpha n} {n \choose i}\right),\\ 
&\geq n - \log_2\left(2^{n H\left(\alpha\right)}\right),\quad \textrm{(Using Lemma \ref{Comb Sum Lemma})}\\
&=n-n H(\alpha).
\end{align*}
Hence for this parametrization of $m$, $\textit{IC}_{\textit{unif}}\left(M_C\right)\in\Omega\left(n\right)$.

\newpage
\section{Proof of Theorem 3}

\begin{theorem*} 
Suppose $m=\alpha n$, $0<\alpha<\frac{1}{2}$ and Alice can abort with probability at most $\delta>0$ on each pair of inputs, $(\vec{x},y)$. Any one-way, public coin, classical strategy, $M_C$, for the exclusion game such that when Alice does not abort, Bob is able to produce $\vec{z}_y\neq\mathcal{M}_y\left(\vec{x}\right)$ for any $y$, has communication cost $\Omega(n)$.
\end{theorem*}
More formally, for $m=\alpha n$, $0<\alpha<\frac{1}{2}$, $CCC_{\delta}^{\rightarrow}\left(\textit{EG}\right)\in\Omega\left(n\right)$.

\begin{proof}
Since Bob has to answer correctly with probability one (on a non-abort message from Alice), we can assume that Bob's strategy is deterministic (by fixing Bob's private coins). Recall that $M_C$ includes the public coins of the protocol and note that Alice is allowed to use private coins. 

Using Lemma \ref{Abort Lemma}, it suffices to show that if $\vec{x}$ and $y$ are chosen independently and from the uniform distribution, $\mu=\textit{unif}$, any strategy that aborts with probability at most $\delta$ but allows Bob to answer correctly otherwise, is such that $\textit{IC}_{\textit{unif}}\left(M_C\right)\in\Omega\left(n\right)$.

Consider $H\left(X|M_C\right)$. Using Eq. \eqref{Class Cond Ent}:
\begin{equation}
H\left(X|M_C\right)=p\left(\text{abort}\right)H\left(X|\textrm{Alice aborts}\right)+p\left(\text{non abort}\right)H\left(X|\textrm{Alice does not abort}\right).
\end{equation}
To obtain a bound on $\textit{IC}_{\textit{unif}}\left(\pi\right)$, we need to upper bound this quantity. The first conditional entropy in the sum is trivially upper bounded by $n$ as $H\left(S|T\right)\leq H(S)$. If Alice does not abort, then, for any input $y$, Bob must win the game with certainty. This means that we can apply the reasoning from the proof of Part 1 of Theorem \ref{Classical IR} to upper bound the second conditional entropy by $\log_2{\gamma_m}$.

This gives:
\begin{align*}
\textit{IC}_{\textit{unif}}\left(M_C\right)&\geq n - \left(\delta n+\left(1-\delta\right)\log_2{\gamma_m}\right),\\
&=\left(1-\delta\right)\left(n-\log_2{\gamma_m}\right).
\end{align*}
From the proof of Part 2b of Theorem \ref{Classical IR}, we know that for $m=\alpha n$, $0<\alpha<\frac{1}{2}$, this expression is $\Omega\left(n\right)$.
\end{proof}

\newpage
\section{Proof of Theorem 4}

\begin{theorem*} %\label{Entanglement Assisted Theorem}
Suppose $m=\alpha n$, $0<\alpha\leq1$ and Alice can abort with probability at most $\delta>0$ on each pair of inputs, $(\vec{x},y)$. Then there exists a one-way, entanglement assisted, strategy for the exclusion game using $\log_{2}k$ bits of communication, such that when Alice does not abort, Bob is able to produce $\vec{z}_y\neq\mathcal{M}_y\left(\vec{x}\right)$ for any $y$. Here $k$ is some constant that depends on $\delta$ but not on $n$.
\end{theorem*}
More formally, $\textit{ECC}_{\delta}\left(\textit{EG}\right)\leq\log_2 k$ for some constant $k$.
\begin{proof}
Making use of the state targeting strategy of \cite{Rudolph2004}, suppose Alice and Bob share $k$ sets of $n$ copies of $\ket{\Phi_{AB}}$ where:
\begin{equation} \label{PhiAB}
\ket{\Phi_{AB}}=\sqrt{\frac{1}{2}\left(1+\frac{\cos\theta_m}{1+\sin\theta_m}\right)}\ket{00}+\sqrt{\frac{1}{2}\left(1-\frac{\cos\theta_m}{1+\sin\theta_m}\right)}\ket{11}.
\end{equation}
Figure \ref{Bob Sphere} shows the reduced state of $\ket{\Phi}_{AB}$ on Bob's Bloch sphere. This reduced state is given by:
\begin{equation}
\rho_B=\left(\begin{array}{cc}
\frac{1}{2}\left(1+\frac{\cos\theta_m}{1+\sin\theta_m}\right) &  0 \\
0 & \frac{1}{2}\left(1-\frac{\cos\theta_m}{1+\sin\theta_m}\right)
\end{array}\right).
\end{equation}
\begin{figure}
\centering
\includegraphics[width=0.8\columnwidth]{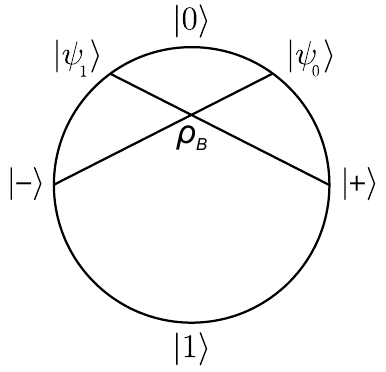}
\caption{$\ket{\Phi}_{AB}$ as viewed on Bob's Bloch sphere \cite{Rudolph2004}. $\rho_B$ denotes his reduced state.} \label{Bob Sphere}
\end{figure}

The following measurement will be useful for our protocol.
\begin{claim}
If Alice measures the state $\ket{\Phi}_{AB}$ with:
\begin{enumerate}
\item $\mathcal{S}=\{{S_0},{S_1}\}$ where:
\begin{equation}
\begin{split} \label{S}
S_0&=\ketbra{s_0}{s_0}, \textrm{ with } \ket{s_0}=\sqrt{\frac{1}{2}\left(1+\frac{\cos\theta_m}{1+\sin\theta_m}\right)}\ket{0}+\sqrt{\frac{1}{2}\left(1-\frac{\cos\theta_m}{1+\sin\theta_m}\right)}\ket{1},\\
S_1&=\ketbra{s_1}{s_1}, \textrm{ with } \ket{s_1}=\sqrt{\frac{1}{2}\left(1-\frac{\cos\theta_m}{1+\sin\theta_m}\right)}\ket{0}-\sqrt{\frac{1}{2}\left(1+\frac{\cos\theta_m}{1+\sin\theta_m}\right)}\ket{1}.
\end{split}
\end{equation}
If the outcome labeled 0 occurs, Bob's half of $\ket{\Phi}_{AB}$ is steered to $\ket{\psi_0\left(\theta_m\right)}$. This happens with probability $\frac{1}{1+\sin\theta_m}$. If the outcome labeled 1 occurs, Bob's half of $\ket{\Phi}_{AB}$ is steered to $\ket{-}$. This happens with probability $\frac{\sin\theta_m}{1+\sin\theta_m}$.
\item $\mathcal{R}=\{{R_0},{R_1}\}$ where:
\begin{equation}
\begin{split} \label{R}
R_0&=\ketbra{r_0}{r_0}, \textrm{ with } \ket{r_0}=\sqrt{\frac{1}{2}\left(1+\frac{\cos\theta_m}{1+\sin\theta_m}\right)}\ket{0}-\sqrt{\frac{1}{2}\left(1-\frac{\cos\theta_m}{1+\sin\theta_m}\right)}\ket{1},\\
R_1&=\ketbra{r_1}{r_1}, \textrm{ with } \ket{r_1}=\sqrt{\frac{1}{2}\left(1-\frac{\cos\theta_m}{1+\sin\theta_m}\right)}\ket{0}+\sqrt{\frac{1}{2}\left(1+\frac{\cos\theta_m}{1+\sin\theta_m}\right)}\ket{1}.
\end{split}
\end{equation}
If the outcome labeled 0 occurs, Bob's half of $\ket{\Phi}_{AB}$ is steered to $\ket{\psi_1\left(\theta_m\right)}$. This happens with probability $\frac{1}{1+\sin\theta_m}$. If the outcome labeled 1 occurs, Bob's half of $\ket{\Phi}_{AB}$ is steered to $\ket{+}$. This happens with probability $\frac{\sin\theta_m}{1+\sin\theta_m}$.
\end{enumerate}
\end{claim}
\begin{proof}
We give the proof for the measurement $\mathcal{S}$, the result for $\mathcal{R}$ follows similarly. The probabilities for each of the measurement outcomes can be calculated from \cite{Rudolph2004}.

Given a normalized entangled state, $a_0\ket{00}+a_1\ket{11}$, suppose Alice performs a measurement on her half and obtains the outcome associated with the projector $x_0\ket{0}+x_1\ket{1}$ with probability $p$. It is easy to see that Bob's system is steered to the (normalized) state:
\begin{equation}
\frac{1}{\sqrt{p}}\left(x_0 a_0 \ket{0}+ x_1 a_1 \ket{1}\right). 
\end{equation}
Using this together with Eq. \ref{PhiAB}, to prove Part 1 it suffices to show that:
\begin{itemize} 
\item For the projective measurement $\ket{s_0}$:
\begin{align*}
\sqrt{1+\sin{\theta_m}}\frac{1}{2}\left(1+\frac{\cos{\theta_m}}{1+\sin{\theta_m}}\right)&=\frac{1}{2}\sqrt{\frac{\left(1+\sin{\theta_m}+\cos{\theta_m}\right)^2}{1+\sin{\theta_m}}},\\
&=\frac{1}{2}\sqrt{\frac{2+2\sin{\theta_m}+2\cos{\theta_m}+2\sin{\theta_m}\cos{\theta_m}}{1+\sin{\theta_m}}},\\
&=\frac{1}{2}\sqrt{\frac{\left(1+\sin{\theta_m}\right)\left(2+2\cos{\theta_m}\right)}{1+\sin{\theta_m}}},\\
&=\cos\left(\frac{\theta_m}{2}\right),
\end{align*} 
and similarly:
\begin{align*}
\sqrt{1+\sin{\theta_m}}\frac{1}{2}\left(1-\frac{\cos{\theta_m}}{1+\sin{\theta_m}}\right)&=\frac{1}{2}\sqrt{\frac{\left(1+\sin{\theta_m}-\cos{\theta_m}\right)^2}{1+\sin{\theta_m}}},\\
&=\frac{1}{2}\sqrt{\frac{2+2\sin{\theta_m}-2\cos{\theta_m}-2\sin{\theta_m}\cos{\theta_m}}{1+\sin{\theta_m}}},\\
&=\frac{1}{2}\sqrt{\frac{\left(1+\sin{\theta_m}\right)\left(2-2\cos{\theta_m}\right)}{1+\sin{\theta_m}}},\\
&=\sin\left(\frac{\theta_m}{2}\right),
\end{align*}
so Bob is steered to $\ket{\psi_0\left(\theta_m\right)}$.
\item For the projective measurement $\ket{s_1}$:
\begin{align*}
\sqrt{\frac{1+\sin\theta_m}{\sin\theta_m}}\frac{1}{2}\sqrt{1+\frac{\cos\theta_m}{1+\sin\theta_m}}\sqrt{1-\frac{\cos\theta_m}{1+\sin\theta_m}}
&=\frac{1}{2}\sqrt{\frac{1+\sin\theta_m}{\sin\theta_m}}\sqrt{1-\frac{\cos^2\theta_m}{\left(1+\sin\theta_m\right)^2}},\\
&=\frac{1}{2}\sqrt{\frac{1+2\sin\theta_m+\sin^2\theta_m-\cos^2\theta_m}{\sin\theta_m+\sin^2\theta_m}},\\
&=\frac{1}{\sqrt{2}},\\
\end{align*}
so Bob is steered to $\ket{-}$.
\end{itemize}
\end{proof}

We now give an explicit protocol using these sets of $\ket{\Phi}_{AB}$, and based on the strategy in \cite{Rudolph2004}, that requires $\log_2{k}$ bits of classical communication:
\begin{enumerate}
\item Alice receives $\vec{x}$ from the referee.
\item For each of the $k$ sets, on the $i^{\textit{th.}}$ copy of $\ket{\Phi}_{AB}$ in that set:
\begin{enumerate}
\item If $x_i=0$, Alice measures with $\mathcal{S}=\{{S_0},{S_1}\}$. 
If the outcome labeled 0 occurs, Bob's half of $\ket{\Phi}_{AB}$ is steered to $\ket{\psi_0\left(\theta_m\right)}$. If the outcome labeled 1 occurs, Bob's half of $\ket{\Phi}_{AB}$ is steered to $\ket{-}$.
\item If $x_i=1$, Alice measures with $\mathcal{R}=\{{R_0},{R_1}\}$. 
If the outcome labeled 0 occurs, Bob's half of $\ket{\Phi}_{AB}$ is steered to $\ket{\psi_1\left(\theta_m\right)}$. If the outcome labeled 1 occurs, Bob's half of $\ket{\Phi}_{AB}$ is steered to $\ket{+}$.
\end{enumerate}
\item If there is a set in which all of the measurements resulted in the $0$ outcome, Alice sends a classical message of length $\log_2{k}$ to Bob indicating which set it was. Otherwise, in each of the $k$ sets, the measurement outcome $1$ occurs at least once so Alice aborts the game and sends a special `abort' symbol to Bob.
\item If Alice did not abort, Bob now has a set of $n$ states that he knows is in the state, $\ket{\Psi_{\vec{x}}\left(\theta_m\right)}$. He runs steps 4-6 of the protocol given in Theorem \ref{Quantum IR} on this set to output a winning answer.
\end{enumerate}

This protocol uses $\log_2{k}$ bits of communication and allows Bob to always output a winning answer when Alice does not abort. It remains to show that $k$ can be chosen to be a constant if Alice is allowed to abort with probability $\delta$.

\begin{claim}
For $m=\alpha n$, the probability that Alice aborts, $P_{\text{abort}}$, is such that:
\begin{equation}
P_{\text{abort}}\leq\left(1-4^{-\frac{1}{\alpha}}\right)^k.
\end{equation}
\end{claim}
\begin{proof}
For each measurement, the probability that Alice obtains the outcome $0$ is given by \cite{Rudolph2004}:
\begin{equation}
P_{\textrm{steer}}=\frac{1}{1+\sin\theta_m}.
\end{equation}
The probability that all $n$ measurements in a set give outcome $0$ is:
\begin{align*}
P^{\textrm{global}}_{\textrm{steer}}&=\left(\frac{1}{1+\sin\theta_m}\right)^n,\\
&=\left(1+2^{\frac{m-2}{m}}-2^{\frac{m-1}{m}}\right)^n,
\end{align*}
where we have used Eq. (\ref{PBR angle Ap}) for $\theta_m$ and the identity $\sin\left(\arctan x\right)=\frac{2x}{1+x^2}$.

Now:
\begin{align*}
\lim_{n\to\infty}P^{\textrm{global}}_{\textrm{steer}}=&\lim_{n\to\infty}\left(1+2^{\frac{\alpha n-2}{\alpha n}}-2^{\frac{\alpha n -1}{\alpha n}}\right)^n,\\
=&\lim_{n\to\infty}\exp\left[n\ln\left[1+2^{\frac{\alpha n-2}{\alpha n}}-2^{\frac{\alpha n -1}{\alpha n}}\right]\right],\\
=&\exp\lim_{t\to0}\frac{\ln\left[1+2^{1-\frac{2t}{\alpha}}-2^{1-\frac{t}{\alpha}}\right]}{t},\\
\textrm{which, using l'Hopital's rule,}\quad=&\exp\lim_{t\to0}\frac{\frac{2}{\alpha}\ln2\left(-2^{1-\frac{2t}{\alpha}}\right)-\frac{1}{\alpha}\ln2\left(-2^{1-\frac{t}{\alpha}}\right)}{1+2^{1-\frac{2t}{\alpha}}-2^{1-\frac{t}{\alpha}}},\\
=&\exp\left[\frac{2}{\alpha}\ln2-\frac{4}{\alpha}\ln2\right],\\
=&4^{-\frac{1}{\alpha}}.
\end{align*}
and as $P^{\textrm{global}}_{\textrm{steer}}$ is monotonically decreasing in $n$, $P^{\textrm{global}}_{\textrm{steer}}\geq{4^{-\frac{1}{\alpha}}}$.

Finally, Alice aborts if each of the $k$ sets fail to steer globally so:
\begin{equation}
P_{\textrm{abort}}=\left(1-P^{\textrm{global}}_{\textrm{steer}}\right)^k\leq\left(1-4^{-\frac{1}{\alpha}}\right)^k.
\end{equation}
\end{proof}

Hence, by choosing $k$ such that $\left(1-4^{-\frac{1}{\alpha}}\right)^k\leq\delta$, Alice and Bob succeed through sending a constant amount of classical communication, regardless of the value of $n$.

\end{proof}

\end{document}